\documentclass[3p]{elsarticle}

\usepackage{hyperref}

\journal{arXiv}

\usepackage{amsfonts}
\usepackage{amsmath, amscd, amssymb, bm, amsbsy, epsf}
    \usepackage{enumerate}
    \usepackage{paralist}
\usepackage{bbm, dsfont}
\usepackage{graphicx,color}
\usepackage{subfigure}
\usepackage{mathrsfs}
\usepackage{verbatim}
\usepackage{inputenc}
\usepackage{diagbox}
\usepackage{algorithm}
\usepackage{algorithmicx}
\usepackage{algpseudocode}
\usepackage{hyperref}
\usepackage{pifont}
\usepackage{inputenc}
\usepackage{bbm, dsfont}
\usepackage{multicol}
\usepackage{balance}
\usepackage{verbatim, booktabs}

\usepackage{multirow, url}
\usepackage[table]{xcolor}
\usepackage{marvosym}
\usepackage{tikz}

\usepackage{cmbright}

\DeclareMathAlphabet{\mathpzc}{OT1}{pzc}{m}{it}
\DeclareMathAlphabet{\mathpzc}{OT1}{pzc}{m}{it}

\newtheorem{theorem}{Theorem}

\newtheorem{proposition}[theorem]{Proposition}
\newdefinition{definition}{Definition}
\newdefinition{hypothesis}{Hypothesis}
\newdefinition{problem}{Problem}
\newdefinition{remark}{Remark}
\newproof{proof}{Proof}
\newdefinition{example}{Example}

\def\defining{\overset{\mathbf{def}}=}

\def\S{\mathcal{S}}

\def\z{\boldsymbol{z} }

\def\z{\mathbb{Z}}
\def\n{\mathbb{N}}
\def\defining{\overset{\mathbf{def}}=}
\def\S{\mathcal{S}}
\def\r{\mathbb{R}}  

\DeclareMathAlphabet{\mathpzc}{OT1}{pzc}{m}{it}
\begin{document}

\begin{frontmatter}

\title{The RaPID-$ \Omega $ system: \\
Room and Proctor Intelligent Decider \\
for large scale tests programming}
\tnotetext[mytitlenote]{This material is based upon work supported by project HERMES 41491 from Universidad Nacional de Colombia,
Sede Medell\'in.}

\author[mymainaddress]{Fernando A Morales} 
\cortext[mycorrespondingauthor]{Corresponding Author}
\ead{famoralesj@unal.edu.co}




\address[mymainaddress]{Escuela de Matem\'aticas
Universidad Nacional de Colombia, Sede Medell\'in \\
Carrera 65 \# 59A--110, Bloque 43, of 106,
Medell\'in - Colombia}


\begin{abstract}
We present the mathematical modeling for the problem of choosing rooms and proctoring crews for massive tests, together with its implementation as the open-box system RaPID\Lightning$ \Omega $. The mathematical model is a binary integer programming problem: a combination of the 0-1 Knapsack problem and the job-assignment problem. The model makes decisions according the following criteria in order of priority: minimization of labor-hours, maximization of equity in the distribution of duties and maximization of the proctoring quality. The software is a digital solution for the aforementioned problem, which is a common need in educational institutions offering large, coordinated, lower-division courses.  The system can be downloaded from  \cite{RaPID}: 
\begin{center}
	\url{https://sites.google.com/a/unal.edu.co/fernando-a-morales-j/home/research/software}
\end{center}  
\end{abstract}


\end{frontmatter}



%
%
%
%
%
%
%
%
%
\section{Introduction}
%
%
In this work we present the mathematical modeling and documentation of the open-box system RaPID\Lightning$ \Omega $, designed to optimize the rooms choice and the proctor scheduling for the logistics of massive tests. This problem is frequent in educational institutions having large coordinated courses (specially for the lower division ones) with simultaneous common tests. Here, we propose a novel mathematical model for this problem and implement a digital solution. Its structure is highly modular, which makes it easy for a programmer to make modifications, it is flexible due to the defined data sets and it is efficient. The software was constructed based on the needs of a specific case: \textit{Escuela de Matem\'aticas} (School of Mathematics) at \textit{Universidad Nacional de Colombia, Sede Medell\'in} (National University of Colombia at Medell\'in). Hence, the simple model examples we present are based on our study case.   

The School of Mathematics is part of the College of Science within the National University of Colombia at Medell\'in, it teaches two types of courses: specialization (advanced undergraduate and graduate courses in mathematics) and service courses (lower division) for the whole University. The latter are: \textit{Differential Calculus} (DC), \textit{Integral Calculus} (IC), \textit{Vector Calculus} (VC), \textit{Differential Equations} (ODE), \textit{Vector \& Analytic Geometry} (VAG), \textit{Linear Algebra} (LA), \textit{Numerical Methods} (NM), \textit{Discrete Mathematics} (DM), \textit{Applied Mathematics} (AM) and \textit{Basic Mathematics} (BM, college algebra). The total demand of these courses amounts to an average of 7200 enrollment registrations per semester. The last three courses, DM, AM, BM, do not test their students in a coordinated fashion but independently (i.e., each lecturer designs his/her own evaluation method); consequently they will not be subject to this analysis. Given that most of the students attending the National University of Colombia at Medell\'in pursue degrees in Engineering, the courses DC, IC, VC, ODE, VAG, LA and MN are massive and they pose significant logistic challenges for booking their respective evaluations; see Table \ref{Tb Historical Enrollment Table} below.
\begin{table}[h!]
	\caption{Historical Enrollment Table}\label{Tb Historical Enrollment Table}
	\def\arraystretch{1.4}
	\normalsize{
		\begin{center}
			\rowcolors{2}{gray!25}{white}
			\begin{tabular}{ c c c c c c c c c }
				\hline
				\rowcolor{gray!50}
				Semester
				& DC
				& IC
				& VC
				& VAG
				& LA
				& ODE
				& NM 
				& Total\\ 
				\hline
				2010--1 &	1631 &	782	  & 381	& 1089	 & 983	& 668 &	 142 & 5676 \\
				2010--2 &	1299 &	1150 &	427	& 1003  & 1007 & 562 &	 261 &	5709 \\
				2011--1 &	1271 &	1136 &	512	& 1078	& 900	& 663 &	 269	& 5829 \\
				2011--2 &	951	  &    850 &  513 &	652	   & 812  &	1170 &	289 &	5237 \\
				2012--1 &	1619 &	1096 &	559	& 1110	& 1116 & 752 &	366	 & 6618 \\
				2012--2 &	1486 &	1190 &	601	& 1076	& 1144	& 825 & 356 &	6678 \\
				2013--1 &	1476 &	1044 &	604	& 1231	& 1037 & 902 &	319 & 6613 \\
				2013--2 &	1446 &	1212 &	549	& 1187	& 1103 & 786 &	326 & 6609 \\
				2014--1 &	1460 &	1184 &	676	& 1192	& 1000	& 890 & 295 &	6697 \\
				2014--2 &	1399 &	1126 &	564	& 1198	& 1012	& 695 &	234	& 6228 \\
				2015--1 &	1097 &	925	  & 565	& 1076	 & 793	&  601	& 201	& 5258 \\
				2015--2 &	1797 &	1214 &	605	& 1314  & 1099 & 808 & 274 & 7111 \\
				2016--1 &	1675 &	1323 &	582	& 1549	& 1017 & 950 & 263 & 7359 \\
				2016--2 &	1569 &	1296 &	594	& 1355	& 1009 & 1019 &	 284 & 7126 \\
				2017--1 &	1513 &	1315 &	515	& 1088	& 798	& 736 & 134	& 6099 \\
				\hline
				Mean & 1445.9	& 1122.9 & 549.8 & 1146.5 & 988.7 & 801.8 & 267.5 & 6323.1 \\
				\hline
			\end{tabular}
		\end{center}
	}
\end{table}
On a typical semester these courses are divided in sections (between 8 and 22, depending on the enrollment) of sizes ranging from 80 to 140 (because of classroom seat capacities). The evaluation consists in three exams which the students take simultaneously, the personnel in charge of proctoring duties consists of approximately 45 lecturers among tenured and adjunct faculty, as well as 70 teaching assistants among graduate and undergraduate students. Moreover the Teaching Assistants and Adjunct Faculty are not full-time employees ergo, they introduce significant time constraints in the task assignment, due to their schedule. Typically, each of the coordinated courses takes three tests during the semester, therefore three rounds of tests need to be scheduled each semester. Currently, each round's selection of rooms and proctoring duties assignment is decided with the RaPID\Lightning$ \Omega $ system.

In contrast with black-box commercial software, this open-box tool is aimed to be easily used and/or modified to the needs of other programming scenarios (most likely other educational institutions with different but similar testing procedures). It shares the spirit (and inspiration) of other open-box systems such as \cite{CarstensenElasticity} and \cite{Bahriawati05threematlab}. We tackle the problem in three steps. First, we model the rooms' optimal choice problem with the Knapsack Problem (see \cite{kellerer2005knapsack} and \cite{martello1990knapsack}), here we minimize the total number of necessary proctors. Second, we model the problem of choosing proctoring crews as the Job Assignment Problem (see \cite{Conforti}), this is done maximizing the equity in the proctoring hours among the personnel. Third, the chose crew is organized in order to maximize the proctoring quality conditions of each room, according to the experience record of each proctor. This is done with a Greedy Algorithm (see \cite{AlgorithmsCormen}), which accommodates the proctors in that order of priority. Each of the steps previously mentioned has a corresponding module inside the system RaPID\Lightning$ \Omega $, these are: \textbf{Room\_Decision.py},  \textbf{Personnel\_Decision.py} and \textbf{Crew\_Organization.py}.  The system is implemented in Python 3.4, it uses libraries such as pandas (Python Data Analysis Library) and SciPy. RaPID\Lightning$ \Omega $ runs from command line and it can be freely downloaded from \cite{RaPID}:
\begin{center}
	\url{https://sites.google.com/a/unal.edu.co/fernando-a-morales-j/home/research/software}
\end{center}  
Three versions are available: a Windows, a Linux-Mac and a Collab version.

The rest of the paper is organized as follows: Section \ref{Sec Room Decision} exposes the modeling and algorithms for the first module of the system concerned about optimal room choice and number of students in each room. Section \ref{Sec Personnel Decision} presents the modeling algorithms of the second step regarding the choice of proctors among available personnel. Section \ref{Sec Crew Organization} presents the algorithm to decide the optimal position to proctor a test, starting from a previously scheduled crew. In Section \ref{Sec Input Datasets} the input datasets are exposed, together with its meaning and its format, in Section \ref{Sec Output Files} the output files are presented, two of these files are final while two are intermediate. Section \ref{Sec Execution and Problems} gives brief directions of execution \& problems and Section \ref{Sec  Conclusions} summarizes the conclusions. The sections  \ref{Sec Room Decision}, \ref{Sec Personnel Decision} and \ref{Sec Crew Organization} are of mathematical interest and for a developer who wants to make adjustments to the current system. On the other hand, a mere user needs to understand only sections \ref{Sec Input Datasets}, \ref{Sec Output Files} \ref{Sec Execution and Problems}. We close this section presenting a Flux Diagram of the RaPIPD\Lightning$ \Omega $-sysem, see Figure \ref{Fig Flux Diagram}.
\begin{figure}[h]
	\centering
	\begin{tikzpicture}
	[scale=.6,auto=left,every node/.style={}]
	\node (n1) at (1, -1)  {$\begin{Bmatrix}
		\small\text{Input DB}\\
		\texttt{Available\_Rooms.xls}
		\end{Bmatrix}$}; 
	\node (n2) at (8, -1)  {$\begin{Bmatrix}
		\small\text{Input DB}\\
		\texttt{Room\_Data.xls}
		\end{Bmatrix}$}; 
	\node (n4) at (1, -4)  {$\begin{bmatrix}
		\texttt{Room\_Decision.py}\\
		\small\text{Proctor hours minimization} \\
		\small\text{Dynamic Programming}
		\end{bmatrix}$};
	\node (n5) at (7.5, -7)  {$\begin{Bmatrix}
		\small\text{Output and imput DB} \\
		\texttt{Scheduled\_Rooms.xls}
		\end{Bmatrix}$}; 
	\node (n6) at (1,-10)  {$\begin{bmatrix}
		\texttt{Personnel\_Decision.py}\\
		\small\text{Equity maximization}\\
		\small\text{Job Assignment Problem}
		\end{bmatrix}$};
	\node (n8) at (-7, -13)  {$\begin{pmatrix}
		\footnotesize\text{Output DB} \\
		\small\texttt{New\_Proctor\_Log.xls}\\
		(\small\text{updated version})
		\end{pmatrix}$}; 
	\node (n9) at (7.5, -13)  {$\begin{Bmatrix}
		\small\text{Output and Input DB} \\
		\texttt{Scheduled\_Crew.xls}
		\end{Bmatrix}$}; 
	\node (n10) at (1,-16)  {$\begin{bmatrix}
		\texttt{Crew\_Organization.py}\\
		\small\text{Optimizaing Proctoring Quality} \\
		\small\text{Greedy Algorithm}
		\end{bmatrix}$};
	\node (n11) at (-7, -1)  {$\begin{Bmatrix}
		\small\text{Input DB}\\
		\texttt{Personnel\_Time.xls}
		\end{Bmatrix}$}; 
	\node (n12) at (-7, -7.50)  {$\begin{Bmatrix}
		\small\text{Input DB}\\
		\texttt{Proctor\_Log.xls}
		\end{Bmatrix}$}; 
	\node (n13) at (1, -19)  {$\begin{pmatrix}
		\small\text{Output DB} \\
		\texttt{Proposed\_Programming.xls}
		\end{pmatrix}$}; 
	
	\node (n16) at (-7, -4.5)  {$\begin{Bmatrix}
		\small\text{Input DB}\\
		\texttt{Professors.xls}
		\end{Bmatrix}$}; 
	
	\foreach \from/\to in {n1/n4, n2/n4, 
		n4/n5, 
		n5/n6, 
		n6/n8, n6/n9, n11/n6, n12/n6, n16/n6,
		n9/n10,
		n10/n13,
		n5/n10}
	\draw[->] (\from) -- (\to);
	
	\foreach \from/\to in {n4/n6, n6/n10}
	\draw[very thick, ->>] (\from) -- (\to);
	\end{tikzpicture}
	\caption{Flux diagram RaPID\Lightning$\Omega$. 
		We use curly brackets to indicate input data bases, square brackets for the algorithms and in round brackets the data bases to be saved for futre and/or direct use.   
	}
	\label{Fig Flux Diagram}
\end{figure}
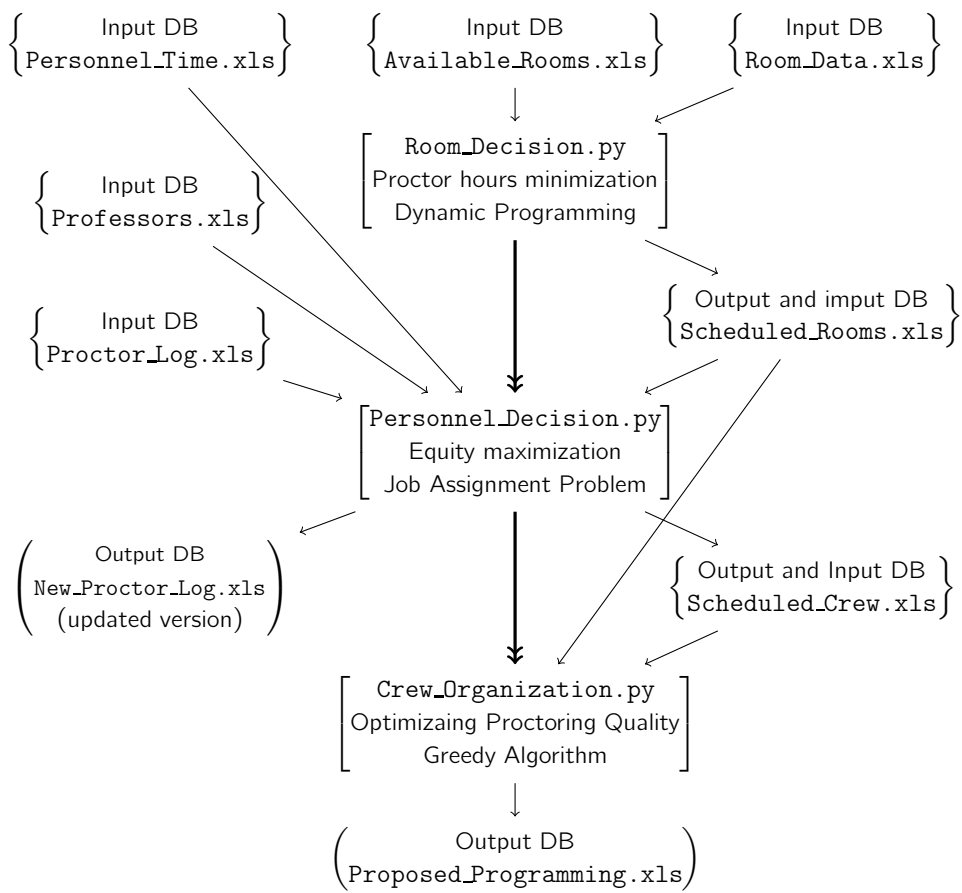
%
%
%
%
%
\section{The Room Decision Problem: module Room\_Decision.py}\label{Sec Room Decision}
%
%
%
%
The first module we are to attack is:
choosing for each test, a set of rooms that minimizes the number of needed proctors. This process is done independently for each test, because there is never time conflict between scheduled tests.
therefore the rooms are not exchangeable items; it consists in two steps. First, the choice of rooms, this is modeled with a 0-1 knapsack problem and solved with dynamic programming algorithm, which is implemented in algorithm \ref{Alg Rooms Decision Algorithm}. Second, the optimal use of the slack between students and available seats (which usually is nonzero), which is solved with a greedy algorithm, which is implemented in algorithm \ref{Alg Slack Distribution Algorithm}. 
%
%
%
%
\subsection{Modeling Choice of Rooms}\label{Sec Rooms Choice}
%
%
In this subsection we explain how we model the rooms' selection. Before starting, two values need to be introduced
\begin{definition}\label{Def Capacity and Weight}
	Let $ \big(c_{i}: 1 \leq i \leq N \big) $ be the list of capacities (quantity of seats) of the available rooms for a given test, define
	\begin{equation}\label{Eqn Capacity and Weight}
	\begin{split}
	& r \in \n \text{ the student-proctor rate},\\
	& w_{i} \defining \big\lceil \frac{c_{i}}{r} \big\rceil \text{ the weight/cost of each room}.
	\end{split}
	\end{equation}
	Here, it is understood that the label $ i = 1, 2, \ldots, N $ stands for each available room and that the ceiling function $ x\mapsto \lceil x \rceil\defining \min \{n \in \z: n\geq x\} $, assigns to any real number $ x \in \r $ the minimum integer greater or equal than $ x $. The student-proctor rate $ r $, is the number of students that an individual must proctor (the system has \text{54 as default value}). 
\end{definition}
Therefore, the problem of choosing rooms is modeled by the 0-1 minimization Knapsack problem
\begin{problem}\label{Pbl Room Decision Problem}
	Let $ \big(c_{i}: 1 \leq i \leq N \big) $ be the list of capacities of the available rooms for a given test and let $ D $ (the demand) be the number of students taking the test. Let $ (w_{i}: 1\leq i \leq N) $ be the list of weights for each room as introduced in Definition \ref{Def Capacity and Weight} then, the problem of minimizing proctors is given by
	\begin{subequations}\label{Eqn Room Decision Problem}
		\begin{equation}\label{Eqn Room Decision Problem Objective Function}
		\min \sum\limits_{i \, = \, 1}^{N} w_{i} x_{i} .
		\end{equation}
		Subject to
		\begin{align}\label{Eqn Room Decision Problem Capacity Constraint}
		& \sum\limits_{i \, = \, 1}^{N} c_{i} x_{i}  \geq D , &
		& x_{i} \in \{0,1\}, &
		& \text{for all } i = 1, 2, \ldots, N .
		\end{align}
		%
	\end{subequations}
	Here, for each $ i = 1, \ldots, N $, the binary variable $ x_{i} $ indicates whether the room $ i $ is chosen ($ x_{i} = 1 $) or not ($ x_{i} = 0 $).
\end{problem}
%
%
Observe that the solution of \textit{Problem} \ref{Pbl Room Decision Problem} above can be found using the solution of the following Knapsack Problem
\begin{problem}\label{Pblm Equivalent Knapsack Problem}
	%
	%
	\begin{subequations}\label{Eqn Equivalent Knapsack Problem}
		\begin{equation}\label{Eqn Equivalent Knapsack Problem Objective Function}
		\max \sum\limits_{i\, \in \, [N] } p_{i}y_{i} ,
		\end{equation}
		subject to
		\begin{equation}\label{Eqn Equivalent Knapsack Problem Capacity Constraint}
		\sum\limits_{i\, \in \, [N] } c_{i}y_{i} \leq \sum\limits_{i\, \in \, [N]}c_{i} -  D,
		\end{equation}
		\begin{align}\label{Eqn Equivalent Knapsack Problem Choice Constraint}
		& y_{i} \in \{0,1\}, &
		& \text{for all } i \,\in \,[N] .
		\end{align}	
	\end{subequations}
	%
	%
\end{problem}
\begin{proposition}\label{Thm Equivalence Knapsack Integer}
	Let $ \boldsymbol{y} \defining \big(y_{i}: i\in [N] \big) \in \{0, 1\}^{N} $ be a solution to \textit{Problem} \ref{Pblm Equivalent Knapsack Problem} and define $ x_{i} \defining 1 - y_{i} $ for all $ i\in [N] $ then, the vector $ \boldsymbol{x} = \big(x_{i}: i\in [N] \big)  \in \{0, 1\}^{N} $ is a solution to \textit{Problem} \ref{Pbl Room Decision Problem}.
\end{proposition}
\begin{proof}
	The proof uses the well-known classic transformation of complementary binary variables, $ x_{i} = 1 - \xi_{i} \in \{0, 1\} $ for all $ i\in [N] $, to relate the problems \ref{Pbl Room Decision Problem} and \ref{Pblm Equivalent Knapsack Problem} (see \textit{Section 13.3.3} in \cite{kellerer2005knapsack} for details).
\end{proof}
There are several ways for solving the 0-1 Knapsack problem \ref{Pblm Equivalent Knapsack Problem}. In the RaPID\Lightning$ \Omega $ system the problem is solved using the technique of dynamic programming (see \textit{Section 11.3} in \cite{Bertsimas} or \textit{Section 2.3} in \cite{kellerer2005knapsack} for details), whose computational complexity is given by
\begin{theorem}[Dynamic Programming Complexity]\label{Thm Dynamic Programming Complexity}
	The 0-1 Knapsack problem can be solved in time $\mathcal{O}(n^{2}p_{\max}) $.
\end{theorem}
\begin{proof}
	See Theorem 11.1 in \cite{Bertsimas}.
\end{proof}
Finally, the implementation is given by the algorithm \ref{Alg Rooms Decision Algorithm}.
\begin{algorithm} 
	\caption{Room Decision Algorithm, decides the choice of rooms and the slack distribution once an optimal solution to Problem \ref{Pbl Room Decision Problem} is found.}
	\label{Alg Rooms Decision Algorithm}
	\begin{algorithmic}[1]
		\Procedure{Room Decision}{Available\_Rooms.xls file, Student-Proctor Rate: $ r $.\newline
			\textbf{User Decision}: Student-Proctor rate $ r $
		}
		\State \textbf{create} the Excel book Scheduled\_Rooms.xls
		\For{column of \textbf{Available\_Rooms.xls} }
		\Comment{Each column is a test, e.g., Table \ref{Tb Available Rooms}}
		\State \textbf{create} the sheet corresponding to the test.
		\State \textbf{retrieve} from \textbf{Available\_Rooms.xls} the information: Rooms' List, Capacities: $ (c_{i})_{i = 1}^{N} $ and Demand: $ D $ corresponding to the test.
		\State \textbf{call dynamic programming solver} (Input: $ \{(c_{i})_{i = 1}^{N} , \sum_{i = 1}^{N}c_{i} - D\} $, Output $ (\xi_{i})_{i = 1}^{N} $  )
		\State \textbf{compute} $ x_{i} \defining \xi_{i} $ for $ i =1, \ldots, N $ 
		\State \textbf{call} Algorithm \ref{Alg Slack Distribution Algorithm} (Input $\{ (c_{i})_{i = 1}^{N}, (x_{i = 1}^{N}) , D  \} $, Output: $ (E_{i})_{i = 1}^{N}$ quantity of students in each chosen room)
		\State \textbf{save} $ (x_{i})_{i = 1}^{N}, (E_{i})_{i = 1}^{N},  (w_{i})_{i = 1}^{N}  $ in the sheet corresponding to the test together with the remaining information displayed in Table \ref{Tb Scheduled Rooms}.
		\EndFor
		\State \textbf{save} book Scheduled\_Rooms.xls
		\EndProcedure
	\end{algorithmic}
\end{algorithm}
%
%
%
%
\subsection{The Slack between available seats and students}\label{Sec Seats-Students Slack}
%
%
For most instances of Problem \ref{Pbl Room Decision Problem}, the optimal solution $ (x_{i})_{i = 1}^{N} $ will not satisfy actively the constraint \eqref{Eqn Room Decision Problem Capacity Constraint} i.e., there will be a slack between the number of available seats ($ \sum_{i\, = \,1}^{N}c_{i}x_{i} $) and the students taking the test ($ D $). The distribution of the slack among the chosen rooms gives rise to a new optimization process. For instance, suppose there are only two rooms available with of 55 seats each, 108 students and the student-proctor rate is $ r = 54 $. Then, the problem \ref{Pbl Room Decision Problem} will choose both rooms i.e., $ x_{1} = x_{2} = 1 $. Next, the greedy algorithm will program 54 students in each room in order to need 2 proctors, instead of programming 53, 55 which would demand 3 proctors. Essentially, the algorithm \ref{Alg Slack Distribution Algorithm}, tries to optimize the slack (from the previously attained solution) in order to reduce the number of proctors once more. To present the algorithm we define the following parameters
\begin{definition}\label{Def Slack and Sorting}
	Let $ \big(c_{i}: 1 \leq i \leq N \big) $ be the list of capacities of the available rooms for a given test, let $ D $ be the number of students taking the test and let $ r $ be the student-proctor rate. Let $ (x_{i})_{i = 1}^{N} $ be an optimal solution of Problem \ref{Pbl Room Decision Problem}. 
	\begin{enumerate}[(i)]
		\item We say that the total slack is given by $ S \defining \sum\limits_{i = 1}^{N} c_{i}x_{i} - D $.
		
		\item For each room, define the list of slack priority coefficients by 
		\begin{align}\label{Eqn Slack Priority Coefficients}
		& s_{i} \equiv c_{i} \mod r ,& 
		& \text{with } 0\leq s_{i} < r ,
		\end{align}
		for all $ i = 1, \ldots, N $. Where it is understood that $ c_{i} = q_{i} r + s_{i} $ with $ 0 \leq s_{i} < r $ being the remainder output in the Euclid's Division Algorithm (see Section 3.3 in \cite{Johnsonbaugh} for details).
	\end{enumerate}
\end{definition}
With the definitions above the, greedy algorithm \ref{Alg Slack Distribution Algorithm} below is implemented. 
\begin{algorithm}[h]
	\caption{Optimal Slack Distribution Algorithm, decides the slack distribution once an optimal solution $ (x_{i})_{i = 1}^{N} $ to Problem \ref{Pbl Room Decision Problem} is found. }
	\label{Alg Slack Distribution Algorithm}
	\begin{algorithmic}[1]
		\Procedure{Slack Distribution}{Capacities: $ (c_{i})_{i = 1}^{N} $, Demand: $ D $, Student-Proctor Rate: $ r $
		}
		
		\State \textbf{compute} list of specific slack priority coefficients $ (s_{i})_{i = 1}^{N} $ \Comment{Introduced in Definition \ref{Def Slack and Sorting}.}
		
		\State \textbf{sort} the list $  (s_{i})_{i = 1}^{N} $ in ascending order
		
		\State \textbf{denote} by $ \sigma \in \S[N] $ the associated ordering permutation, i.e., 
		\begin{align}\label{Eqn Sorting Permutation}
		& s_{\sigma(i)} \leq s_{\sigma(i + 1)}, &
		& \text{for all } i = 1, \ldots, N -1 .
		\end{align}

		\State $ Slack \defining \sum_{i = 1}^{N} c_{i} - D $ \Comment{Initializing the available slack}
		\State $ e_{i} \defining c_{i} $ for all $ i = 1, \ldots, N $ \Comment{Initializing the located students in each room (overcounted)}
		
		\For {$ i =1, \ldots, N $} 
		\If {$ (0 < Slack) $ and $ (Slack \leq s_{\sigma(i)}) $}
		\State $ e_{\sigma(i)} = e_{\sigma(i)} - Slack $
		\ElsIf { $ (Slack > s_{\sigma(i)}) $ }
		\State $ e_{\sigma(i)} = e_{\sigma(i)} - s_{\sigma(i)} $
		\State $ Slack = Slack - (s_{\sigma(i)} + 1) $ 
		\EndIf
		\State \textbf{compute} $ w_{i} \defining \big\lceil \frac{e_{i}}{r} \big\rceil  $. \Comment{Updated number of needed proctors}
		\EndFor 
		\State \Return $ (e_{i})_{i = 1}^{N}, (w_{i})_{i = 1}^{N} $		
		\EndProcedure
	\end{algorithmic}
\end{algorithm}
\begin{remark}\label{Rem Slack Distribution Algorithm}
	The slack priority coefficient, introduced in Definition \ref{Def Slack and Sorting} essentially seeks in which rooms a proctor can be reduced by removing the minimum possible number of programmed students. Namely, if $ t = 54 $ a room with 55 students can be proctored by only one TA, by removing only one student. If the room has 56, two students should be removed to attain the previous result, which is a stronger requirement in terms of available slack. In the same reasoning, it is better to remove one student from a room of $ 109 $ ($ s = 1 $) than 2 from a room of $ 56 $ ($ s = 2$).
\end{remark}
\begin{theorem}\label{Thm Computational Complexity Greedy Rooms}
	The computational complexity of Algorithm \ref{Alg Slack Distribution Algorithm} is $ \mathcal{O}(N \log_{2} N) $, where $ N $ is the total number of chosen rooms. 
\end{theorem}
\begin{proof}
	This is a Greedy Algorithm based on sorting. It is known that the standard MergeSort algorithm does the job in $ \mathcal{O}(N \log_{2} N) $, see Section 19.2.2 in \cite{BonaWalk}
\end{proof}
%
%
\section{The Personnel Selection Problem: module Personnel\_Decision.py}\label{Sec Personnel Decision}
%
%
%
%
In this section we present the mathematical model for the problem of choosing personnel for proctoring shifts with a sense of equity. The input of this problem can be seen in Figure \ref{Fig Flux Diagram}: $ \bullet $ The number of tests, rooms and quantity of TA's needed to proctor (contained in the file Scheduled\_Rooms.xls) $ \bullet  $ The time availability for each TA $ \bullet $ The lecturers directly assigned to proctor (Professors.xls) $ \bullet $ The record of previous service of the each TA in a time window of interest (Proctor\_Log.xls). Also notice that in this particular case all the TAs are exchangeable (unlike the lecturers) therefore, the problem of choosing a proctoring crew will be done simultaneously for all the tests in the round, in contrast with the test-wise design model of \textit{Section} \ref{Sec Room Decision}. This section is divided in two parts, the presentation of the mathematical model and the presentation of the algorithm.
%
%
%
%
\subsection{The mathematical model}\label{Sec Personnnel Decision model}
%
%
In this section we derive a mathematical model for the problem of selecting personnel for proctoring duties (or shifts). It will be seen that the model is the Job Assignment problem with constraints. Before we can introduce it, some definitions are in order
\begin{definition}\label{Def Personnel Decision Variables}
	Let $ T $ be the number of tests in the round and let $ P $ be the total number of part-time employed TAs. 
	\begin{enumerate}[(i)]
		\item For each test $ t $, let $ N = N(t) $ be the number of scheduled rooms for the test and let $ (w_{i}^{(t)})_{i = 1}^{N(t)} $ be the list of needed proctors in each booked room. Then, the total number of proctors needed in the test is given by
		\begin{equation}\label{Eq Total Proctors per Test}
		W^{(t)} \defining \sum\limits_{i = 1}^{N(t)} w_{i}^{(t)} .
		\end{equation}

		\item For each $ t = 1, \ldots, T $ and $ p = 1, \ldots , P $, let $ y_{p}^{(t)} \in \{0,1\} $ be the decision variable defined by
		\begin{equation}\label{Eq Decison Variable per Test and TA}
		y_{p}^{(t)} \defining \begin{cases}
		1,  & \text{if proctor } p \text{ is assigned to test t }, \\
		0,  & \text{otherwise.}
		\end{cases}
		\end{equation}
		%
		Also define the availability coefficient $ a_{p}^{(t)} \in \{0, 1\} $ by
		\begin{equation}\label{Eq Availability per Test and TA}
		a_{p}^{(t)} \defining \begin{cases}
		1, & \text{if proctor } p \text{ is available at the time of test } t, \\
		0, & \text{otherwise.}
		\end{cases}
		\end{equation}

		\item For each $ p $ (identifying a TA), denote by $ L_{p} $ the number of served shifts in the Proctor\_Log.xls file (i.e., the TA's service record before running the algorithm). 
		
		\item The global service average $ \alpha $, after the current round of tests, is computed in the natural way, i.e.,
		\begin{equation}\label{Eqn Service Average}
		\alpha \defining \frac{1}{P}\Big(\sum_{p \, = \, 1}^{P} L_{p} + \sum_{t \, = \, 1 }^{T}W^{(t)} \Big) .
		\end{equation}
	\end{enumerate}
\end{definition}
\begin{remark}
	Observe that due to the expression \eqref{Eqn Service Average} the quantity $ \alpha $ is known. On the other hand, observe that $ \sum_{t = 1}^{T} y_{p}^{(t)} $ quantifies the total number of shifts that each TA has in the current round of tests. The alternative definition $ \alpha =  \frac{1}{P}\sum_{p \, = \, 1}^{P} (L_{p} + \sum_{t = 1}^{T} y_{p}^{(t)} ) $, would be mathematically equivalent to \eqref{Eqn Service Average}, but it would (unnecessarily)  include $ \alpha $ as a variable.
\end{remark}
With the definitions above, it is direct to see that we want to solve the following problem
\begin{problem}\label{Pbl Original MInimization Job Assignment Problem}
	\begin{subequations}\label{Eqn Original Job Assignment Problem}
		\begin{equation}\label{Eqn Original Job Assignment Objective Function}
		\min \max_{p \, = \, 1}^{P}\Big\vert \sum\limits_{t\, = \, 1 }^{T} y_{p}^{(t)} + L_{p} - \alpha \Big\vert .
		\end{equation}
		Subject to
		\begin{align}\label{Eqn Original Job Assignment Necessary Proctors Constraints}
		& \sum_{p \, = \, 1}^{P} a_{p}^{(t)}y_{p}^{(t)} = W^{(t)} , &
		& \text{for all } t = 1, \ldots, T.
		\end{align}
		\begin{align}\label{Eqn Original Job Assignment Choice Constraint}
		&  y_{p}^{(t)} \in \{0,1\}, &
		& \text{for all } t = 1, \ldots T \text{ and } p = 1, \ldots, P .
		\end{align}	
	\end{subequations}
\end{problem}
The problem \ref{Pbl Original MInimization Job Assignment Problem} is clearly not a linear programming problem. In order to transform it, we introduce a new continuous variable $ z $ (not necessarily integer), verifying the constraints $ \big\vert \sum_{t\, = \, 1 }^{T} y_{p}^{(t)} + L_{p} - \alpha \big\vert  \leq z $ for all $ p = 1\ldots, P $. Clearly, the absolute sense can be decoupled in two inequality constraints and given that this is an integer programming problem, the following refinement is introduced
\begin{equation}\label{Eq Proctors Service-Equity Constraints}
\begin{split}
\sum\limits_{t\, = \, 1 }^{T} y_{p}^{(t)} + L_{p} - \lceil\alpha\rceil 
\leq \sum\limits_{t\, = \, 1 }^{T} y_{p}^{(t)} + L_{p} - \alpha 
\leq z, \\
- z \leq \sum\limits_{t\, = \, 1 }^{T} y_{p}^{(t)} + L_{p} - \alpha 
\leq \sum\limits_{t\, = \, 1 }^{T} y_{p}^{(t)} + L_{p} - \lfloor\alpha\rfloor .
\end{split}
\end{equation}
With the introduction of the variable $ z $ above and the natural linear relaxation, the problem \ref{Pbl Original MInimization Job Assignment Problem} can be reformulated as the following linear optimization problem. 
\begin{problem}\label{Pbl Job Assignment Problem}
	With the variables and quantities introduced in Definition \ref{Def Personnel Decision Variables}, the problem of assigning TAs for proctoring duties as fairly as possible, considering their time constraints, is modeled by the problem 
	\begin{subequations}\label{Eqn Job Assignment Problem}
		\begin{equation}\label{Eqn Job Assignment Objective Function}
		\min z .
		\end{equation}
		Subject to
		\begin{equation}\label{Eqn Job Assignment Absolute Sense Constraints}
		\begin{split}
		& \sum\limits_{t\, = \, 1 }^{T} y_{p}^{(t)} - z \leq - L_{p} + \lceil\alpha\rceil , \\
		& \sum\limits_{t\, = \, 1 }^{T} y_{p}^{(t)}  - z \leq L_{p} - \lfloor\alpha\rfloor,  \quad
		\forall p = 1, \ldots , P.
		\end{split}
		\end{equation}
		\begin{align}\label{Eqn Job Assignment Necessary Proctors Constraints}
		& \sum_{p \, = \, 1}^{P} a_{p}^{(t)}y_{p}^{(t)} = W^{(t)} , &
		& \text{for all } t = 1, \ldots, T.
		\end{align}
		%
		\begin{align}\label{Eqn Job Assignment Choice LR Constraint}
		&  0 \leq y_{p}^{(t)} \leq 1 , &
		& \text{for all } t = 1, \ldots , T \text{ and } p = 1, \ldots, P .
		\end{align}	
	\end{subequations}
\end{problem}
\begin{theorem}
	For the job assignment problem, the binary integer choice constraint \eqref{Eqn Original Job Assignment Choice Constraint}, can be replaced by the natural linear relaxation constraint \eqref{Eqn Job Assignment Choice LR Constraint} and both problems have the same integer optimal solution.
\end{theorem}
\begin{proof}
	See Chapter 4 in \cite{Conforti}.
\end{proof}
%
%
%
%
\subsection{The Personnel\_Decision.py module}\label{Sec Personnel Decision Algorithm}
%
%
With the exposition above, the algorithm of the module Personnel\_Decision.py is summarized in the pseudocode \ref{Alg Personnel Decision Algorithm}.
\begin{algorithm}[h]
	\caption{Personnel Decision Algorithm, decides the proctoring crew for all the tests in the round.}
	\label{Alg Personnel Decision Algorithm}
	\begin{algorithmic}[1]
		\Procedure{Personnnel Decision}{Scheduled\_Rooms.xls file, Personnel\_Time.xls, Proctor\_Log.xls, Profesors.xls.}
		\State \textbf{load} the information corresponding to the input files.
		\State \textbf{include} the lecturers in the file Scheduled\_Crew.xls, assigned to their corresponding courses.
		\State \textbf{compute} the number of necessary proctors for each test. \Comment{Coefficients ($ W^{(t)}, t = 1, \ldots, T $), Equation, \eqref{Eq Total Proctors per Test}.}
		\State \textbf{account for} the time constraint coefficients for each TA. \Comment{Coefficientes ($ a_{p}^{(t)}, t = 1, \ldots, T, p = 1, \ldots, P $), Equation \eqref{Eq Availability per Test and TA}.}
		\State \textbf{include} the proctors necessity constraints \Comment{Equation \eqref{Eqn Job Assignment Necessary Proctors Constraints}.}
		\State \textbf{compute} the service average value \Comment{Coefficient $ \alpha $, Equation \eqref{Eqn Service Average}.}
		\State \textbf{include} the proctors service-equity constraints \Comment{Equation \eqref{Eqn Job Assignment Absolute Sense Constraints}.}
		\State \textbf{solve} the Problem \ref{Pbl Job Assignment Problem} $ \leftarrow $ \textbf{call} a linear program solver \Comment{The system uses scipy.optimize.linprog}
		\State \textbf{create and save} the Updated\_Proctor\_Log.xls file.
		\State \textbf{create and save} the Scheduled\_Crew.xls file.
		\EndProcedure
	\end{algorithmic}
\end{algorithm}
\begin{remark}[Computational Complexity]\label{Rem Computational Complexity Personnel Decision}
	The algorithm \ref{Alg Personnel Decision Algorithm} solves the problem using the Simplex method. It is know that the Simplex algorithm can grow exponentially with respect to the number of constraints, for some particular cases (see Theorem 3.5 in \cite{Bertsimas}). However, it has been observed in the practice that the method typically takes $ \mathcal{O}(m) $ iterations and $ \mathcal{O}(m n) $ operations per iteration. Where $ m $ is the number of constraints and $ n $ is the dimension of the vector. In the particular case of RaPID\Lightning$ \Omega $, it can be safely assumed that this is the computational complexity of  Algorithm \ref{Alg Personnel Decision Algorithm}. Given that $ n = T \times P + 1 $ and $ m = 2T + P $, it follows that the algorithm runs in time
	\begin{equation}\label{Eq Computational Complexity Personnel Decision}
	\mathcal{O}\big( (2T + P )^{2} (T + P)\big),
	\end{equation}
	i.e., it is polynomial in the number of tests and the number of proctors.
\end{remark}
%
%
%
%
%
%
\section{The Crew Organization Problem: module Crew\_Organization.py}\label{Sec Crew Organization}
%
%

Once the proctoring crew is decided in the previous steps, the final step is to distribute the TAs in strategic positions in order to maximize the proctoring quality. More specifically, it is clear that rooms with only one proctor assigned should get the more experienced proctors, while a mixed experienced-unexperienced couple should be assigned to rooms needing two proctors. In the same fashion, the supervisor should be an experienced undergraduate TA, as it is a position with higher responsibility. At this point the proctors are no longer exchangeable between tests, therefore the process is done independently as in \textit{Section \ref{Sec Room Decision}}. Although, it is possible to construct a mathematical integer programming model for this stage, it is more practical to use a greedy algorithm for making these decisions (see \cite{Dasgupta} for a comprehensive exposition on greedy algorithms).
\footnotesize
\begin{algorithm}[h!]
	\caption{Crew Organization Algorithm, decides proctoring positions.}
	\label{Alg Crew Organization Algorithm}
	\begin{algorithmic}[1]
		\Procedure{Crew Organization}{Scheduled\_Rooms.xls, Scheduled\_Crew.xls
		}
		\State \textbf{load} the sheet Scheduled\_Crew.xls as Data Frame
		\State \textbf{load} the Excel book Scheduled\_Rooms.xls as Data Frame
		\State \textbf{create} column ``Num\_Level" in Scheduled\_Rooms.xls (3 Undergraduate, 2 Graduate, 1 PhD)
		\State \textbf{create} the Excel book Proposed\_Programming.xls
		\For{test\_idx in \textbf{Scheduled\_Rooms.xls} } \Comment{Each sheet in the book is a test.}
		\State \textbf{create} sheet test\_idx in the book Proposed\_Programming.xls
		\State \textbf{select} the rooms scheduled for test\_idx, i.e., Scheduled\_Rooms[test = test\_idx].
		\State \textbf{create} the Local\_Test\_Frame, with columns: Room, Proctors, Position, Students \Comment{The column ``Position" indicates if it is the 1st, 2nd, 3rd... proctor in the room.}
		\State \textbf{define} $ns \defining $ the number of needed supervisors.
		\State \textbf{create} Local\_Proctors\_Frame $ \leftarrow $ Scheduled\_Crew[Test = test\_idx] \Comment{The proctors assigned to the test.}
		\State \textbf{sort} lexicographically the Local\_Proctors\_Frame with the priority: [``Num\_Level", ``Experience"] in the order [Descending, Descending].
		\State \textbf{select} Undergraduate\_Local\_Proctors\_Frame $ \leftarrow $ Local\_Proctors\_Frame[Level = Undergraduate]
		\Comment{Choose the undergraduate proctors from the local scheduled proctors frame.}
		\State \textbf{assign} the first $ ns $ proctors from Undergraduate\_Local\_Proctors\_Frame as supervisors. \State \textbf{save} the supervisors in the sheet test\_idx 
		\State \textbf{remove} the chosen supervisors from Local\_Proctors\_Frame.  
		\State \textbf{remove} the supervisors' rows from Local\_Test\_Frame. 
		\State \textbf{sort} lexicographically the Local\_Test\_Frame with the priority: [``Proctors", ``Position", ``Students"] in the orders [Ascending, Ascending, Descending]. 
		\State \textbf{paste} Local\_Test\_Frame with Local\_Test\_Frame and \textbf{define} the outcome as Local\_Programming. 
		\State \textbf{save} Local\_Programming in the sheet test\_idx of the book Proposed\_Programming.xls.
		\EndFor
		\State \textbf{save} book Proposed\_Programming.xls
		\EndProcedure
	\end{algorithmic}
\end{algorithm}
\normalsize
The greedy algorithm \ref{Alg Crew Organization Algorithm} is based on lexicographic sorting of the scheduled crew's data. More specifically a pairing (match) has to be done: the proctoring positions and the proctors themselves. It is direct to see that the proctors should be sorted according to their academic level in the first place (3 Undergraduate, 2 Graduate, 1 Lecturer) and to their experience in the second place (see Table \ref{Tb Personnel Time}). The proctoring positions need further explanation. Observe that any test will need a number of supervisors (according to the enrollment size) and a number of proctors in each room. The process of selecting supervisors and the process of selecting proctoring positions are different. Hence, two separate steps will be done for completing the tasks. 

With the observations above, the Crew\_Organization.py module works with the algorithm \ref{Alg Crew Organization Algorithm} below.
\begin{table}[h!]
	\begin{center}
		\normalsize{
			\rowcolors{2}{gray!25}{white}
			\caption{Example of Scheduled Rooms Remark \ref{Rem Crew Organization}}\label{Tb Sorting Rooms Start}
			\begin{tabular}{ c c c }
				\hline
				\rowcolor{gray!50}
				Room & Proctors  & Students \\
				\hline
				16-223 & 2 & 63 \\
				46-209 & 1 & 50 \\
				46-307 & 2 & 80 \\
				Supervisor 1 & & \\
				\hline
			\end{tabular}
		}
	\end{center}
\end{table}
\begin{table}[h!]
	\begin{center}
		\normalsize{
			\rowcolors{2}{gray!25}{white}
			\caption{Example of Sorted Rooms Remark \ref{Rem Crew Organization}}\label{Tb Sorting Rooms}
			\begin{tabular}{ c c c c }
				\hline
				\rowcolor{gray!50}
				Room & Proctors & Position & Students \\
				\hline
				46-209 & 1 & 1 & 50 \\
				46-307 & 2 & 1 & 80 \\
				16-223 & 2 & 1 & 63 \\
				46-307 & 2 & 2 & 80 \\
				16-223 & 2 & 2 & 63 \\
				Supervisor 1 & & & \\
				\hline
			\end{tabular}
		}
	\end{center}
\end{table}
\begin{remark}\label{Rem Crew Organization}
	Some clarifications are necessary for a deeper understanding of the design and purpose of the algorithm \ref{Alg Crew Organization Algorithm}.
	\begin{enumerate}[(i)]
		\item To understand the line 9 in Algorithm \ref{Alg Crew Organization Algorithm} consider the example of Table \ref{Tb Sorting Rooms Start}. One supervisor is needed, the room 46-209 needs one proctor while the rooms 16-223 and 46-307 need two. Hence, a data frame including positions, has to be created; see Table \ref{Tb Sorting Rooms}.
		
		\item The command of line 18 in Algorithm \ref{Alg Crew Organization Algorithm} has the following motivation. Consider the example of Table \ref{Tb Sorting Rooms Start}. Then, line 16 in Algorithm \ref{Alg Crew Organization Algorithm} would deliver the table \ref{Tb Sorting Rooms} below. It is clear that once the scheduled proctors are sorted as indicated in line 10 of the algorithm \ref{Alg Crew Organization Algorithm}, the most experienced proctors will be distributed, as evenly as possible through the rooms; taking first the rows with one proctor first. The next rows will be those of 2 proctors with higher number of students having label ``1" in the column ``Position", here will be assigned the not so experienced proctors. Next, the less experienced proctors will be assigned to the rows having ``2" in the column ``Position", and so forth. 
	\end{enumerate}
\end{remark}
\begin{theorem}\label{Thm Computational Complexity Greedy Crew}
	The computational complexity of Algorithm \ref{Alg Crew Organization Algorithm} is $ \mathcal{O}(P \log_{2} P) $, where $ P $ is the maximum number of proctors that a test in the round has. 
\end{theorem}
\begin{proof}
	This is a Greedy Algorithm based on sorting. It is known that the standard MergeSort algorithm does the job in $ \mathcal{O}(N \log_{2} N) $, see Section 19.2.2 in \cite{BonaWalk}
\end{proof}
%
%
%
%
\section{The Input Datasets}\label{Sec Input Datasets}
%
%
In the present section we describe the input data files for the RaPID\Lightning$ \Omega $ system, explain the contents and structure of each dataset, as well as its motivation. In our study case, a round of tests for the seven courses is typically scheduled in the span of two weeks. The input files are five: Available\_Rooms.xls, Room\_Data.xls, Personnel\_Time.xls, Proctor\_Log.xls and Professors.xls. The first two are concerned with physical spaces, while the last three accounting for human resources.
%
%
%
%
\subsection{The Available\_Rooms.xls file}\label{Sec Available Rooms}
%
%
The first input data is the spreadsheet \textbf{Available\_Rooms.xls}, which contains the structure presented in the example of Table \ref{Tb Available Rooms} below. For the School of Mathematics, the first step in programming an examination is to request the necessary seats for each activity to the University's \textit{Office of Building Management}. The information sent is summarized in the last three rows of Table \ref{Tb Available Rooms}. Conversely, the Office of Building Management replies a list of rooms available at the time and date requested and capable of holding the necessary enrollment, see Table \ref{Tb Available Rooms} for an example. 
\begin{table}[h!]
	\caption{Example of Available Rooms}\label{Tb Available Rooms}
	\def\arraystretch{1.2}
	\footnotesize{
		\begin{center}
			\rowcolors{2}{gray!25}{white}
			\begin{tabular}{ c c c c c c c c }
				\hline
				\rowcolor{gray!50}
				Room
				& DC
				& IC
				& VC
				& VAG
				& LA
				& ODE
				& NM 
				\\ 
				\hline
				03-210	& 30 & & & & & & \\						
				04-108	& 29 & & & & & & \\							
				04-109	& 60 & & & & & & \\							
				04-110	& & & &	58 & & & \\			
				04-111	& & & & & &  & \\							
				04-206  & 65 & & & & & & \\							
				05-101	& & & & & & & 79 \\
				11-102	& & & & & & 56 &	\\
				11-124	& 22 & 22 & & & & & \\						
				11-125	& &	41 & & & & & \\						
				11-202	& & & & & & 150 & \\	
				11-203	& & 39 &	& 39 & & 43 & \\	
				11-208	& 50 & 50 &	& 50 & & & \\			
				11-209	& 30& 30 &	& 30 & & & \\			
				11-225	& &	60	& &	60 & &	60 & \\	
				14-109	& 51 & 49 & & & & & \\					
				14-232	& 50 & & & 50 & & & \\			
				16-223	& 63 & 63 &	63 & 63 & 63 & & \\		
				16-224	& 60 & 72 & 60 & 72	& 60 & 58 & \\	
				21-303	& 34 & & & 34 & & & \\			
				21-307	& 52 & & 52 & & 52 & & \\		
				21-314	& 79 & 79 &	79 & 70 & 79 & & \\		
				21-320	& 64 & 64 &64 & 64	& 64 & 58 & \\	
				21-328	& & & 30 & & & & \\			
				21-331	& 29 & & & & & & \\						
				24-307	& & & & & & 170 & \\	
				25-301	& & & & & &	90 & 79 \\
				41-102	& & 106 & & & & & \\					
				41-103	& 106 & & 106 & & 106 &	106 & 102\\
				43-110	& & & &	54 & & 53 & \\	
				43-111	& & & & & & 53 & \\	
				46-114	& & & &	56 & & & \\			
				46-208	& 47 & 47 & & 47 & & & \\		
				46-209	& 50 & 50 & 50 & & 50 & & \\		
				46-210	& 52 & & 52 & & 52 & & \\	
				46-211	& 52 & 52 & & 52 & & & \\			
				46-212	& 49 & & & 49 & & & \\			
				46-301	& 44 & 44 &	& 44 & & & \\			
				46-303	& &	43 & & 43 & & & \\			
				46-304	& 41 &	41 & & 41 & & & \\			
				46-307	& 100 &	100	& 100 & & 100 & & \\		
				46-311	& &	41	& & & & & \\				
				\hline	 
				Students & 1300 & 1050 & 608 & 951 & 600 & 822 & 150 \\						
				Date &	30-III & 01-IV &	06-IV &	01-IV &	06-IV &	02-IV &	05-IV \\
				Time &	Sa 12-14 &	Mo 08-10 &	Sa 14-16 &	Mo 10-12 &	Sa 12-14 &	Mo 10-12 &	Mo 08-10 \\
				\hline
			\end{tabular}
		\end{center}
	}
\end{table}
It should be noticed that there is a slack for all the cases i.e., the total capacity of the available rooms always exceeds the number of students: DC 1354, IC 1093, VC 626, VAG 1006, LA 626, ODE 897, MN 402. Consequently, the choice of available rooms can be done so to minimize the number of necessary proctors. 
\begin{remark}[Format Available\_Rooms.xls]\label{Rem Available Rooms}
	Some format guidelines  \textbf{must be observed} in the Available\_Rooms.xls file for the correct functioning of RaPID\Lightning$ \Omega $.
	\begin{enumerate}[(i)]
		
		
		\item The time must include the abbreviation day, a blank space, two digits for each hour and a hyphen in between: \textbf{dd TT-TT}, e.g., \textbf{Mo 08-10} instead of \textbf{Mo 8-10}. Of course, the time slots can be modified according to the scheduling needs (e.g. Monday from 9:00 to 12:00), as long as the format \textbf{dd TT-TT}, is consistently preserved through the datasets (e.g. the file Personnel\_Time.xls in Section \ref{Sec Personnel Time}).
		
		\item The names of the courses (acronyms or not) in the columns can be modified as long as the labels are consistent with those in the file Professors.xls. Naturally, the columns of the file can be increased or decreased according to necessity.
		
		\item The data in the column ``Room", need not be sorted.
		
		\item The capacity of each room need not be written on the courses columns, it suffices to write ``1", to indicate that the room is available for the activity. The system will actually read the value of the room capacity from the file Room\_Data.xls (see Section \ref{Sec Room Data}); however, it may be desirable to write the capacity in this file, in order to check on the input spreadsheet itself, if the number of available seats is greater or equal than the number of students. 
	\end{enumerate} 
\end{remark}
%
%
%
%
\subsection{The Room\_Data.xls file}\label{Sec Room Data}
%
%
The Room\_Data.xls sheet centralizes the information about all the classrooms on campus, not only those needed for the examination activities, see Table \ref{Tb Room Data} for a minimal example. 
\begin{table}[h!]
	\caption{Example of Room Data}\label{Tb Room Data}
	\def\arraystretch{1.2}
	\normalsize{
		\begin{center}
			\rowcolors{2}{gray!25}{white}
			\begin{tabular}{ c c c }
				\hline
				\rowcolor{gray!50}
				Room &	Capacity & Observations\\
				\hline
				03-210 & 32 & Doorkeeper	\\
				04-108 & 28	& Key \\	
				04-109 & 70 & Access Code \\		
				04-207 & 70	& Wheelchair Ramp \\	
				05-101 & 80	& Doorkeeper \\	
				11-102 & 56	& Card	
				\\ 
				\hline
			\end{tabular}
		\end{center}
	}
\end{table}

The table contains the number of seats or capacity of each room and a column of observations where some annotations can be made such as: how is the room to be opened or if it has accommodations for students with disabilities. This file is the most stable of all, as it changes only when the nature of the rooms change, therefore it is less vulnerable to human error than the file Available\_Rooms.xls; this is why RaPID\Lightning$ \Omega $ reads the capacities from this file.      
%
%
%
%
\subsection{The Personnel\_Time.xls file}\label{Sec Personnel Time}
%
%
The third input data for the system is the spreadsheet \textbf{Personnel\_Time.xls}, it has the structure presented in the minimal example of Table \ref{Tb Personnel Time} below. As explained in the introduction the team of proctors is not made of full-time employees, therefore there are time constraints when scheduling a test. This is particularly acute in the case of the Teaching Assistants whose labor duties amount to 10 hours per week and their academic duties may be time-conflicting with the examination activities. 
\begin{table}[h!]
	\caption{Example of Time Personnel}\label{Tb Personnel Time}
	\def\arraystretch{1.2}
	\normalsize{
		\begin{center}
			\rowcolors{2}{gray!25}{white}
			\begin{tabular}{ c c c c c c c c c}
				\hline
				\rowcolor{gray!50}
				Name &	Cell &	email &	ID	 & Experience &	
				Level  &	Mo 10-12 & Mo 12-14 & Sa 08-10\\
				\hline
				TA 1 &	 C 1	& 1@m.co & ID 1 &	1 &	Undergraduate & Day Off &	& 1 \\ 
				TA 2 &	 C 2	& 2@m.co & ID 2 &	2 & Undergraduate & Busy & 1 & 1 \\ 
				TA 3 &   C 3	& 3@m.co & ID 3 &	1 & Undergraduate & 1 & Class & 1 \\ 
				TA 4 &	 C 4	& 4@m.co & ID 4 &	1 &	Undergraduate &	1 & 1 & 1 \\
				TA 5 &	 C 5	& 5@m.co & ID 5 & 2 & Undergraduate & 1 & NA & 1 \\ 
				TA 6 &	 C 6	& 6@m.co & ID 6 & 2 &	Postgraduate &	1 & 0 & 1 
				\\ 
				\hline
			\end{tabular}
		\end{center}
	}
\end{table}

The table contains fields for identification (``Name" and ``ID"), contact (``Cell", ``email") and rating (``Experience" and ``Level"). Finally, the availability fields are represented by multiple time slots; in this minimal example only three time-slots were included: Mo 08-10,  Mo 10-12, Sa 08-10. In practice, all possible time-slots should be contained in the table, namely: Mo 08-10, Mo 10-12, ..., Mo 16-18, Td 08-10, ..., Fr 16-18, Sa 08-10, ..., Sa 16-18. 
\begin{remark}[Format Personnel\_Time.xls]\label{Rem Personnel Time}
	The following instructions must be observed when building the Personnel\_Time.xls sheet.
	\begin{enumerate}[(i)]
		
		\item The time slots \textbf{dd TT-TT}, must be consistent with those of the file Available\_Rooms.xls (see Remark \ref{Rem Available Rooms} (iii) Section \ref{Sec Available Rooms}) for the system to work properly. They can be modified according to the scheduling needs (e.g. Monday from 9:00 to 12:00), as long as the consistency between datasets and the format \textbf{dd XX-YY} are preserved.
		
		\item The time slots columns \textbf{dd TT-TT} must indicate whether or not an individual is available. To indicate availability use ``1". 
		
		\item Only \textbf{availability} indicated by the number ``1" is important for later calculations. In particular, if any other information is set (e.g. ``Available", ``Free", ``Busy", etc.) the system will understand that the individual is \textbf{unavailable} at that time slot.
		
		\item There is no need to indicate unavailability. As shown in the example, the reason why a TA is not available, can be declared  or not. Unavailability can also be marked with a ``0" as in the example above. However, this annotations will not impact on the system. 
		
		\item The column ``Level" has to be filled with the words ``Undergraduate" or ``Postgraduate", for the system to understand the academic level of each TA. Furthermore, this information will play a key role in the greedy algorithms of the module Crew\_Organization.xls, see Section \ref{Sec Crew Organization} and Algorithm \ref{Alg Crew Organization Algorithm}. 
		
	\end{enumerate}
\end{remark}
%
%
%
%
%
%
\subsection{The Proctor\_Log.xls file}\label{Sec Proctor Log}
%
%
The \textbf{Proctor\_Log.xls} is a file containing the record of proctoring duties that the TAs have served in an observation time-window namely: a term, a semester or a year (depending on the institutional policy). Its structure is presented in the minimal example of Table \ref{Tb Proctor Log} below and it is fairly similar to that of Personnel\_Time.xls; it agrees on the columns holding each TA's information. Each of the remaining columns represent an examination activity that took place in the observation time-window, up to the programming date. For each examination event, the number ``1" indicates that the individual served on it. The column ``Total" indicates the total number of shifts the TA has taken so far. 
\begin{remark}\label{Rem Proctor Log}
	Some observations about this file are the following
	\begin{enumerate}[(i)]
		\item When the system is initiated for the first time this file has to contain all the columns holding the information of each TA (``Name", ``Cell", ``email", ``ID", ``Experience", ``Level") and the column ``Total" with value ``0" in all its rows. 
		
		\item Once the system is executed for the next round of tests, an updated file: \textbf{Updated\_Proctor\_Log.xls} will be generated automatically, see Section \ref{Sec Updated Proctor Log} and Table \ref{Tb Updated Proctor Log} for this file.
	\end{enumerate}
\end{remark}
\begin{table}[h!]
	\caption{Example of Proctor Log}\label{Tb Proctor Log}
	\def\arraystretch{1.2}
	\normalsize{
		\begin{center}
			\rowcolors{2}{gray!25}{white}
			\begin{tabular}{ c c c c c c c c}
				\hline
				\rowcolor{gray!50}
				Name &	Cell &	email &	ID	 & Experience &	
				Level  &	ODE, 04-II	
				& Total \\
				\hline
				TA 1 &	 C 1 &	1@m.co  & ID 1 & 1 & Undergraduate & & 0 \\ 
				TA 2 &	 C 2 & 2@m.co  & ID 2 & 2 & Undergraduate & 1 & 1 \\ 
				TA 3 &	 C 3 &	3@m.co	& ID 3	& 1	& Undergraduate	& 1 & 1 \\ 
				TA 4 &	C 4 & 4@m.co	& ID 4	& 1	& Undergraduate	& 1 & 1\\ 
				TA 5 &	C 5 &	5@m.co	& ID 5	& 2	& Undergraduate	& 1	& 1 \\ 
				TA 6 &	C 6 & 6@m.co	& ID 6 & 2	& Postgraduate & & 0 
				\\ 
				\hline
			\end{tabular}
		\end{center}
	}
\end{table}
%
%
%
%
%
%
\subsection{The Professors.xls file}\label{Sec Professors File}
%
%
The input file \textbf{Professors.xls}, is a spreadsheet with the structure of the minimal example presented in Table \ref{Tb Professors} below. Unlike the TAs, the lecturers are full-time employees, therefore they pose no time constraints when scheduled for proctoring duties. The column ``Coordinator" indicates if the Lecturer is the coordinator of the course and will be the general supervisor of the examination activity, therefore he/she will not be scheduled for proctoring testing rooms. The columns ``Subject", ``Subject\_2" indicate which subject are they lecturing in the academic period. 
\begin{remark}\label{Rem Professors File}
	Some observations are 
	\begin{enumerate}[(i)]
		\item Indicating that a faculty member is a course coordinator must be done with the word ``yes". If any other character or word is set (e.g. ``Yes", ``1", ``coordinator") the system \textbf{will not} understand the corresponding individual as coordinator and will include him/her in the proctoring duties as any other lecturer.
		
		\item By default, the system will assign an instructor to proctor only the examination activities of the subject he/she is lecturing. More specifically, in Table \ref{Tb Professors} Lec 3 will be assigned to proctor only examinations of DC. 
		
		\item Some instructors may have two or more service courses assigned, however, the system takes into account only the first subject for proctoring duties, e.g., Lec 4 will be assigned to proctor only examinations of ODE and not those of AL. Therefore, the remaining subjects may be omitted. 
		
		\item If it is the User Institution's policy is to make lecturers participate in the proctoring of every subject they teach, it suffices to create one more row for the second subject. For instance Lec 6 teaches two subjects, therefore two rows should be created for he/she, one having ``VAG" in the column ``Subject", the other having ``DC" in the same column.
		
		\item If the User Institution has the policy of having only the TAs proctoring the tests, the column ``Coordinator" should be filled with the word ``yes". This will suffice to exclude the instructors from the job assignment (as it does with the actual coordinator), but is important to stress that the Professors.xls file \textbf{must exist}, with the \textbf{columns} described above (even if it is empty), for RaPID\Lightning$ \Omega $ to work correctly.
	\end{enumerate}
\end{remark}  
\begin{table}[h!]
	\caption{Example of Professors file}\label{Tb Professors}
	\def\arraystretch{1.2}
	\normalsize{
		\begin{center}
			\rowcolors{2}{gray!25}{white}
			\begin{tabular}{ c c c c c c }
				\hline
				\rowcolor{gray!50}
				Name & Coordinator & Subject & Subject\_2 & Cell & email \\
				\hline
				Lec 1 &	yes	& VC &	&  C 100 & Lec1@m.co \\
				Lec 2 &		& NM &	& C 200 &	Lec2@m.co \\
				Lec 3 &		& DC &	&  C 300 & Lec3@m.co \\
				Lec 4 & 	& ODE & AL & C 400 &	Lec4@m.co \\
				Lec 5 &		& MD &	 & C 500 & Lec5@m.co \\
				Lec 6 &		& VAG &	DC & C 600	& Lec6@m.co
				\\ 
				\hline
			\end{tabular}
		\end{center}
	}
\end{table}
%
%
%
%
%
%
\section{The Output Files}\label{Sec Output Files}
%
%
In the present section we describe the output files that RaPID\Lightning$ \Omega $  produces and explain its contents. From the \textbf{user's} point of view, only two files are important: \textbf{Proposed\_Programming.xls} and \textbf{Updated\_Proctor\_Log.xls}, which will be explained first. However, a \textbf{developer} should understand the other two files which are produced as an intermediate step towards the final solution, from one system's module to the next; these are: \textbf{Scheduled\_Rooms.xls} and \textbf{Sheduled\_Crew.xls}
%
%
%
%
%
%
\subsection{The Proposed\_Programming.xls file}\label{Sec Proposed Programming}
%
%
The \textbf{Proposed\_Programming.xls} file is an excel book and it is the ultimate goal of the system. Here, there is a sheet for each course in the round; in our study case: DC, IC, VC, VAG, LA, ODE and NM. Each sheet has the structure of Table \ref{Tb Proposed Programming}. The fields are Room (code of the room), Envelope (or pack of tests), Observations, Capacity, Students, Slack, Test, Date, Proctors (number of assigned proctors), Name (name of the assigned proctors), Cell and email.
\begin{table}[p]
	\rotatebox{90}{
		\begin{minipage}{\textheight}
			\begin{center}
				\scriptsize{
					\caption{Example of Proposed Programming CV, 608 Students}\label{Tb Proposed Programming}
					\def\arraystretch{1.2}
					\rowcolors{2}{gray!25}{white}
					\begin{tabular}{ c c c c c c c c c c c c }
						\hline
						\rowcolor{gray!50}
						Room &	Envelope &	Observations &	Capacity &	Students &
						Slack &	Test &	Date &	Proctors &	Name &	Cell &	email
						\\
						\hline
						41-103 & 1 & Doorkeeper	& 106	& 106	& 0	& VC & Sa 14-16 06-IV	& 2	& TA 7	& C 7	& 7@m.co \\
						41-103 & 1 & Doorkeeper & 106	& 106	& 0	& VC & Sa 14-16 06-IV	& 2	& TA 41	& C 41 & 41@m.co \\
						46-307 & 2	& Card	& 80 &	80	& 0	& VC & Sa 14-16 06-IV & 2	& TA 20	& C 20 & 20@m.co \\
						46-307 & 2	& Card	& 80 & 80	& 0	& VC & Sa 14-16 06-IV & 2	& TA 48	& C 48 & 48@m.co \\
						46-209 & 3	& Card	& 50 &	50	& 0	& VC & Sa 14-16 06-IV & 1	& TA 65	& C 65 & 65@m.co \\
						16-224	& 4	& Card	& 72 & 57	& 15 &	VC & Sa 14-16 06-IV & 2 &	TA 28 &	C 28 & 28@m.co \\
						16-224	& 4	& Card	& 72 & 57	& 15 & VC & Sa 14-16 06-IV & 2	& TA 63	& C 63 & 63@m.co \\
						21-320	& 5	& Card	& 79 &	79	& 0	& VC & Sa 14-16 06-IV &	2	& TA 23	& C 23 & 23@m.co\\
						21-320	& 5	& Card	& 79 & 79	& 0	& VC & Sa 14-16 06-IV &	2	& TA 55	& C 55 &	55@m.co \\
						21-314	& 6	& Card	& 79 & 79	& 0	& VC & Sa 14-16 06-IV & 2 & TA 24 & C 24 & 24@m.co \\
						21-314	& 6	& Card	& 79 & 79	& 0	& VC & Sa 14-16 06-IV & 2 & TA 59 & C 59	& 59@m.co \\
						46-210	& 7	& Card	& 52 & 52	& 0	& VC & Sa 14-16 06-IV & 1 &	Lec 17 & C 1700 &	Lec17@m.co \\
						21-307	& 8	& Card	& 52 & 52	& 0	& VC & Sa 14-16 06-IV &	1 &	Lec 25	& C 2500	& Lec25@m.co \\
						16-223	& 9	& Card	& 63 & 53	& 10 & VC & Sa 14-16 06-IV & 2 & TA 4 &	C 4	& 4@m.co \\
						16-223	& 9	& Card	& 63 & 53	& 10 & VC & Sa 14-16 06-IV	& 2	& TA 67 & C 67 &	67@m.co\\
						Supervisor 1 &	na & na & na & na & na & VC & Sa 14-16 06-IV &	1 & TA 46 & C 46	& 46@m.co \\
						\hline
					\end{tabular}
				}
				\newline
				\newline
				\footnotesize{
					\caption{Example of Scheduled Rooms CV, 608 Students}\label{Tb Scheduled Rooms}
					\def\arraystretch{1.2}
					\rowcolors{2}{gray!25}{white}
					\begin{tabular}{ c c c c c c c c c }
						\hline
						\rowcolor{gray!50}
						Room &	Envelope &	Proctors &	Observations &	Capacity &	
						Students &	Slack &	Test &	Date \\
						\hline\\
						46-210	& 1	& 1	& Card	& 52 & 52 & 0 & VC & Sa 14-16 06-IV \\
						21-314	& 2	& 2	& Card	& 79 & 79 &	0 & VC & Sa 14-16	06-IV \\
						16-223	& 3	& 2	& Card	& 63 & 53 & 10 & VC & Sa 14-16	06-IV \\
						46-209	& 4	& 1	& Card	& 50 & 50 & 0 & VC & Sa 14-16	06-IV \\
						46-307	& 5	& 2	& Card	& 80 & 80 & 0 & VC & Sa 14-16	06-IV \\
						21-307	& 6	& 1	& Card	& 52 & 52 & 0 & VC & Sa 14-16	06-IV \\
						16-224	& 7	& 2	& Card	& 72 & 57 & 15 & VC & Sa 14-16	06-IV\\
						21-320	& 8	& 2	& Card	& 79 & 79 & 0 & VC & Sa 14-16	06-IV \\
						41-103	& 9	& 2	& Doorkeeper &	106 & 106 & 0 & VC & Sa 14-16 06-IV\\
						Supervisor 1	& na & 1 & na &	na & na & na & VC & Sa 14-16 06-IV \\
						\hline
					\end{tabular}
				}
			\end{center}
		\end{minipage}
	}
\end{table}
\begin{remark}\label{Rem Proposed Programming}
	Two observations are in order 
	\begin{enumerate}[(i)]
		\item The number of assigned proctors to a room depends on the number of students. The default value is one proctor per 54 students, which can be changed by the user at the time of executing RaPID\Lightning$ \Omega $, see Section \ref{Sec Rooms Choice}, Definition \ref{Def Capacity and Weight} to change this value.
		
		\item If two (or more) proctors are assigned to one room, two (or more) rows will be equal except for the name of the proctor (e.g., the first and second row in Table \ref{Tb Proposed Programming}). As a consequence of these repetitions, the sum of the column ``Students" in this file, will not yield the number of students taking the test, as it happens in the example at hand.
	\end{enumerate}
\end{remark} 
%
%
\subsection{The New\_Proctor\_Log.xls file}\label{Sec Updated Proctor Log}
%
%
Once the system is executed for a next round of tests, an updated file \newline
\textbf{New\_Proctor\_Log.xls}, will be generated automatically. This will have the previous service record and it will paste it on the left: one column per examination activity. In each column a ``1" will be written if the individual was selected to serve in the corresponding activity and also the ``Total" column will be updated. 

Consider the minimal example presented in Table \ref{Tb Updated Proctor Log}. Here, it is understood that the proctor log file is that of Table \ref{Tb Proctor Log} (only one examination ODE on February the 4th took place before) and the system is programming only one examination in the next round, which is AVG on March the 4th. The updated log of Table \ref{Tb Updated Proctor Log} writes the proctoring duties for all the employees and an updated ``Total" of service.
\begin{table}[h!]
	\caption{Example of Updated Proctor Log}\label{Tb Updated Proctor Log}
	\def\arraystretch{1.2}
	\normalsize{
		\begin{center}
			\rowcolors{2}{gray!25}{white}
			\begin{tabular}{ c c c c c c c c c}
				\hline
				\rowcolor{gray!50}
				Name &	Cell &	email &	ID	 & Experience &	
				Level  &	ODE, 04-II	& AVG, 04-III
				& Total \\
				\hline
				TA 1 &	C 1 &	1@m.co  & ID 1 & 1 & Undergraduate & & 1 & 1 \\
				TA 2 &	C 2 & 2@m.co  & ID 2 & 2 & Undergraduate & 1 & &	1 \\
				TA 3 &	C 3 &	3@m.co	& ID 3	& 1	& Undergraduate	& 1 & 1 & 2 \\		
				TA 4 &	C 4 & 4@m.co	& ID 4	& 1	& Undergraduate	& 1 &	& 1\\
				TA 5 &	C 5 &	5@m.co	& ID 5	& 2	& Undergraduate	& 1	& 1 & 2 \\
				TA 6 &	C 6 & 6@m.co	& ID 6 & 2	& Postgraduate & & & 0
				\\ 
				\hline
			\end{tabular}
		\end{center}
	}
\end{table}
\begin{remark}\label{Rem Updated Proctor Log}
	The following must be observed
	\begin{enumerate}[(i)]
		\item Typically, the updated version of the proctors' log should increase its columns in more than one. In our study case, the School of Mathematics from Universidad Nacional de Colombia, Sede Medell\'in, the update increases seven columns each round, because that is the number of massive courses that RaPID\Lightning$ \Omega $ manages for the User Institution. 
		
		\item The more courses that are programmed in one round, the more chances for optimization instances. The examination rounds need not be equal as some courses may take two midterms, while others take three.
		
		\item The file New\_Proctor\_Log.xls is created independently from Proctor\_Log.xls, instead of simply overwriting it, for security reasons. It will also be useful for later manual corrections, for instance, some employees selected to proctor may have a license (medical or personal). Due to the random nature of these exceptions, it will be wiser to handle them manually, by a human supervisor, than trying to incorporate them in the system.
		
		\item Once the round of examinations is over, for the next round of tests, the Proctor\_Log.xls file must be replaced by the New\_Proctor\_Log.xls file.
	\end{enumerate}
\end{remark}
%
%
%
%
\subsection{The Scheduled\_Rooms.xls file}\label{Sec Scheduled Rooms}
%
%
The first module of the system \textbf{Room\_Decision.py} processes the file Available\_Rooms.xls (see Table \ref{Tb Available Rooms}) and chooses rooms in order to minimize the number of necessary proctors (see Section \ref{Sec Room Decision} for the exposition of its algorithm). Its results are summarized in the file \textbf{Scheduled\_Rooms.xls}, which is an excel book having one sheet for each programmed test in the round. Each sheet has the structure of Table \ref{Tb Scheduled Rooms}. As it can be observed it is very similar to the Proposed\_Programming.xls file but without the assigned proctors. This is because the Scheduled\_Rooms.xls file is an intermediate step, when only the rooms have been decided. Hence, it is an internal file, moreover it is part of the input data for the module \textbf{Personnel\_Decision.py}, which selects a team of proctors (see Section \ref{Sec Personnel Decision} for the presentation of its algorithm).
%
%
%
%
\subsection{The Scheduled\_Crew.xls file}\label{Sec Scheduled Crew}
%
%
The second module of the system \textbf{Personnel\_Decision.py} processes the input files Personnel\_Time.xls, Proctor\_Log.xls, Professors.xls together with the internal file Scheduled\_Rooms.xls to make job assignment decisions based on fairness. More specifically, the system tries to keep as close as possible, the number of shifts that each TA has in his/her service record. Its decisions are summarized in the file \textbf{Scheduled\_Crew.xls}, a minimal (and incomplete) example can be observed in Table \ref{Tb Scheduled Crew}. 
At this stage, each examination has a crew of proctors which can be gathered/identified by the label \textbf{Test}, indicating which test is going to be proctored by each individual in the file. This file, together with Scheduled\_Rooms.xls are the input data for the third module, \textbf{Crew\_Organization.py} to decide how to organize the previously selected team according to proctoring quality control (see Section \ref{Sec Crew Organization} for the explanation of the algorithm).
\begin{table}[h!]
	\caption{Example of Scheduled Crew file}\label{Tb Scheduled Crew}
	\def\arraystretch{1.2}
	\footnotesize{
		\begin{center}
			\rowcolors{2}{gray!25}{white}
			\begin{tabular}{ c c c c c c }
				\hline
				\rowcolor{gray!50}
				Cell &	Experience	& Level	& Name	& Test & email \\
				\hline
				C 12 & 3 &	Undergraduate &	TA 12	& LA &	12@m.co \\
				C 11 &	1 &	Undergraduate & TA 11	& ODE &	11@m.co \\
				C 10 &	2 &	Undergraduate &	TA 10 &	DC	& 10@m.co \\
				C 10 &	2 &	Undergraduate &	TA 10 &	IC	& 10@m.co \\
				C 800 &	10	& PhD & Lec 8 &	LA & Lec8@m.co \\
				C 700 & 10	& PhD &	Lec 7 &	DC & Lec7@m.co \\
				\hline
			\end{tabular}
		\end{center}
	}
\end{table}
%
%
%
%
\section{Execution and Problems Description}\label{Sec Execution and Problems}
%
%
%
%
%
%
%
\subsection{Execution}\label{Sec Execution}
%
%
In order to run the program notice that the downloaded folder will contain all the necessary files for a \textbf{full example}, these are:
\begin{enumerate}[(i)]
	\item The Python scripts: Room\_Decision.py, Personnel\_Decision.py, Crew\_Organizacion.py, \newline
	Attendance\_Lists.py and RaPID-Omega.py. 
	\item The input files: Available\_Rooms.xls, Room\_Data.xls, Personnel\_Time.xls, Proctor\_Log.xls and Professors.xls.  
	
\end{enumerate}
The program \textbf{runs} from \textbf{command line} on a computer having installed Python 3.4.4 or later using the instruction:
\begin{enumerate}[(i)]
	\item Windows version
	\begin{center}
		\textbf{python.exe RaPID-Omega.py -t 54} (54 states the student-proctor rate).
	\end{center}
	\item Linux-Mac version
	\begin{center}
		\textbf{python3 RaPID-Omega.py -t 54} (54 states the student-proctor rate).
	\end{center}
\end{enumerate}
Once the program is executed the following files are generated: Scheduled\_Rooms.xls, Scheduled\_Crew.xls, Proctor\_Log.xls, New\_Proctor\_Log.xls and Proposed\_Programming.xls
%
%
%
%
%
%
%
%
\subsection{Problems Description and Hints}
%
%

\begin{enumerate}[I.]
	\item \textbf{Troubleshooting}. The system RaPID\Lightning$ \Omega $ may present the following problems when starting it
	\begin{enumerate}[(i)]
		\item The \textbf{python version} installed in your computer is not 3.4. In such case, you need to modify manually the first line of the four codes: Room\_Decision.py, Personnel\_Decision.py, Crew\_Organization and RaPID-Omega.py
		
		\item The \textbf{necessary libraries} are not installed. Verify that pandas, scipy and numpy are installed.
		
		\item RaPID\Lightning$ \Omega $ can be run directly from \textbf{python editors} like IDLE or Spider inside Anaconda. Before trying, make sure it is possible to do it without having to do a conversion format work (this is the case for Jupyter inside Anaconda).
		
		\item \textbf{Hint.} It is highly recommendable that you begin by making run the system in the Full Example folder (includes full data sets), so you can check RaPID\Lightning$ \Omega $ works properly on your computer. 
		
		\item \textbf{Installation.} If you require assistance installing on your computer the necessary software to run RaPID\Lightning$ \Omega $ check the following websites: Python \cite{Python}, Numpy and Scipy \cite{NumpyScipy}, Pandas \cite{Pandas}. A comfortable installer of python libraries is PIP, see \cite{PIP}
	\end{enumerate}

	\item \textbf{Data basis.}
	The RaPID\Lightning$ \Omega $ solution depends critically on the input data basis, therefore, these files are the main source of potential problems which we describe below. 
	\begin{enumerate}[(i)]
		\item \textbf{Consistency between data basis.} A common source of errors is the consistency of the labels used across the several data sets. For instance if the time windows used for the tests do not agree in format with those of the personnel time availability, or that the name of an employee in the data base Personnel\_Time.xls, does not agree with his/her name in the file Registro\_Vigilancias.xls. See the section \ref{Sec Input Datasets} for more details and examples. 
		
		\item \textbf{Repeated data.} It is possible that the file Personnel\_Time.xls and/or Registro\_Vigilancias.xls present two or more repeated names. This, due to human error or mere coincidence (two employees with the same name). This will cause errors in the execution of the program. First, make sure that there are no mistakes in the names written in both data basis. If by coincidence two or more employees have the same name \textbf{differentiate them artificially}, e.g., John I, John II, etc.
		
		\item \textbf{Hint.} When building your own data sets, start from those in this folder, modifying them one by one. This way, you can make partial checks to see if your data sets are constructed correctly. If, the format/structure of the data sets of the Institution/Client at hand is different from those defined in RaPID\Lightning$ \Omega $, it is safer (and presumably easier) to develop an independent module migrating the Institution's formats to the RaPID\Lightning$ \Omega $ formats.
	\end{enumerate}
	
\end{enumerate}
%
%
%
%
%
%
\section{Conclusions}\label{Sec Conclusions}
%
%
The present work delivers the following conclusions
\begin{enumerate}[(i)]
	\item The problem of programming classrooms with proctoring personnel for large scale tests has been analyzed, dividing the process in three sub-problems: the choice of rooms, the choice of proctoring crews and the crew organization.
	
	\item Each of the three problems is mathematically modeled with an integer programming problem and solved independently. Its corresponding digital implementation constitutes a module in the RaPID\Lightning$ \Omega $ system. In addition, the nature of the problem suggests which input data sets to define, their structure, as well as the format of the output given by the system.
	
	\item The model is successful, practical and flexible. Its good mathematical results, robustness and its correct implementation are tested by years of service at the School of Mathematics at the National University of Colombia.
	
	\item The computational complexity of the model is the sum of the complexities of each module. The first module solves the problem with dynamic programming which is solved in pseudo-polynomial time, as stated in Theorem \ref{Thm Dynamic Programming Complexity} and a Greedy Algorithm which also runs in polynomial time, see Theorem \ref{Thm Computational Complexity Greedy Rooms}. The second module uses the Simplex method, the particular usage or RaPID\Lightning$ \Omega $ (as well as literature on the field) has shown a polynomial running time average behavior, see Remark \ref{Rem Computational Complexity Personnel Decision}. The third module uses a mere Greedy Algorithm which runs in polynomial time as discussed in Theorem \ref{Thm Computational Complexity Greedy Crew}. It follows that the whole system runs in polynomial time with respect to the size of the data that is requested to handle.
	
	\item The modular approach to the problem makes it easy for a programmer to introduce modifications according to the needs of the Institution at hand. Its spirit is to stay as an open-box system, furnishing  the algorithmic core treatment of the problem.   
	
	\item The system RaPID\Lightning$ \Omega $ is coded in python 3.4, therefore is completely free as it uses only free tools and libraries, even its input and output excel books and sheets can be handled with their corresponding open versions (Apache, Libre, etc.), because the system does not depend on excel for its computations; it only uses as input and output format.
\end{enumerate}  

%
%
%
%
%
%
\section*{Acknowledgements}
The Author wishes to thank Universidad Nacional de Colombia, Sede Medell\'in for supporting the production of this work through the project Hermes 45713. 

%


\end{document}